\documentclass[11pt]{amsart}
\usepackage{times,amssymb,amsmath,float,euscript,mathpazo,nicefrac,
pstricks,color,mathrsfs}
\usepackage[dvips]{epsfig}

\evensidemargin \oddsidemargin
\newtheorem{theorem}{Theorem}
\newtheorem{lemma}[theorem]{Lemma}

\newtheorem{corollary}[theorem]{Corollary}
\newtheorem{definition}{Definition}

\renewenvironment{proof}[1]{\noindent {\it Proof~:} #1}
{\ \rule{1mm}{2mm}\medskip}

\newcommand{\remove}[1]{}

\newcommand\wt{\mbox{{\rm wt}\,}}
 
\def\wt{\qopname\relax{no}{w}}

\newcommand\nc\newcommand
\nc\bfa{{\boldsymbol a}}\nc\bfA{{\mathbf A}}\nc\cA{{\mathcal A}}
\nc\bfb{{\boldsymbol b}}\nc\bfB{{\mathbf B}}\nc\cB{{\mathcal B}}
\nc\bfc{{\boldsymbol c}}\nc\bfC{{\mathbf C}}\nc\cC{{\EuScript C}}
\nc\bfd{{\boldsymbol d}}\nc\bfD{{\mathbf D}}\nc\cD{{\mathcal D}}\nc\sD{{\mathscr D}}
\nc\bfe{{\boldsymbol e}}\nc\bfE{{\mathbf E}}\nc\cE{{\EuScript E}}
\nc\bff{{\boldsymbol f}}\nc\bfF{{\mathbf F}}\nc\cF{{\mathcal F}}
\nc\bfg{{\boldsymbol g}}\nc\bfG{{\mathbf G}}\nc\cG{{\mathcal G}}
\nc\bfh{{\boldsymbol h}}\nc\bfH{{\mathbf H}}\nc\cH{{\mathcal H}}
\nc\bfi{{\boldsymbol i}}\nc\bfI{{\mathbf I}}\nc\cI{{\mathcal I}}
\nc\bfj{{\boldsymbol j}}\nc\bfJ{{\mathbf J}}\nc\cJ{{\mathcal J}}
\nc\bfk{{\boldsymbol k}}\nc\bfK{{\mathbf K}}\nc\cK{{\mathcal K}}
\nc\bfl{{\boldsymbol l}}\nc\bfL{{\mathbf L}}\nc\cL{{\mathcal L}}
\nc\bfm{{\boldsymbol m}}\nc\bfM{{\mathbf M}}\nc\cM{{\mathcal M}}
\nc\bfn{{\boldsymbol n}}\nc\bfN{{\mathbf N}}\nc\cN{{\mathcal N}}
\nc\bfo{{\boldsymbol o}}\nc\bfO{{\mathbf O}}\nc\cO{{\mathcal O}}
\nc\bfp{{\boldsymbol p}}\nc\bfP{{\mathbf P}}\nc\cP{{\mathcal P}}
\nc\bfq{{\boldsymbol q}}\nc\bfQ{{\mathbf Q}}\nc\cQ{{\mathcal Q}}
\nc\bfr{{\boldsymbol r}}\nc\bfR{{\mathbf R}}\nc\cR{{\mathcal R}}
\nc\bfs{{\boldsymbol s}}\nc\bfS{{\mathbf S}}\nc\cS{{\mathcal S}}
\nc\bft{{\boldsymbol t}}\nc\bfT{{\mathbf T}}\nc\cT{{\mathcal T}}\nc\sT{{\mathscr T}}
\nc\bfu{{\boldsymbol u}}\nc\bfU{{\mathbf U}}\nc\cU{{\mathcal U}}
\nc\bfv{{\boldsymbol v}}\nc\bfV{{\mathbf V}}\nc\cV{{\mathcal V}}
\nc\bfw{{\boldsymbol w}}\nc\bfW{{\mathbf W}}\nc\cW{{\mathcal W}}\nc\sW{{\mathscr W}}
\nc\bfx{{\boldsymbol x}}\nc\bfX{{\mathbf Z}}\nc\cX{{\EuScript X}}
\nc\bfy{{\boldsymbol y}}\nc\bfY{{\mathbf Y}}\nc\cY{{\EuScript Y}}\nc\sY{{\mathscr Y}}
\nc\bfz{{\boldsymbol z}}\nc\bfZ{{\mathbf Z}}\nc\cZ{{\mathcal Z}}\nc\sZ{{\mathscr Z}}
\nc\od{{\bar d}}\nc\ow{{\bar w}}\nc\odelta{{\bar\delta}}
\nc\ox{{\bar x}}\nc\oy{{\bar y}}\nc\ou{{\bar u}}
\nc\oh{{\bar h}}

\newcommand\ee{{\mathbb E}}

\nc\dgv{\delta_{\text{\rm GV}}}
\nc\rcrit{R_{\text{\rm crit}}}
\nc\Esp{E_{\text{\rm sp}}}
\renewcommand\epsilon{\varepsilon}

\newcommand{\beeq}{\begin{eqnarray*}}
\newcommand{\eneq}{\end{eqnarray*}}
\renewcommand{\le}{\leqslant}
\renewcommand{\leq}{\leqslant}
\renewcommand{\ge}{\geqslant}
\renewcommand{\geq}{\geqslant}

\DeclareMathAlphabet{\mathbfsl}{OML}{zpple}{b}{it} 
\newcommand{\bfsl}{\bfseries\slshape}

\begin{document}
\title[Codes on hypergraphs]
{Weight distribution and decoding of codes on hypergraphs} 
\author[A. Barg]{Alexander Barg$^\ast$}
\author[A. Mazumdar]{Arya Mazumdar$^\S$}
\author[G. Z\'emor]{Gilles Z\'emor$^\dag$}
\thanks{$^\ast$ Dept. of ECE and Institute for Systems Research, 
University of Maryland, College Park, MD 20742 and Institute for 
Problems of Information Transmission, Moscow, Russia, Email: abarg@umd.edu.
Research supported in part by NSF grants
CCF0515124 and CCF0635271 and by NSA grant H98230-06-1-0044.}
\thanks{$^\S$ Dept. of ECE, University of Maryland, College Park, MD 20742, Email: arya@umd.edu. Research supported in part by NSF grant CCF0635271.}
\thanks{$^\dag$ Institut de Math\'ematiques de Bordeaux,
Universit\'e de Bordeaux 1,
351 cours de la Lib\'eration,
33405 Talence, France, Email: gilles.zemor@math.u-bordeaux1.fr}

\begin{abstract}
Codes on hypergraphs are an extension of the well-studied family of codes
on bipartite graphs. 
Bilu and Hoory (2004) constructed an explicit family of codes on regular 
$t$-partite hypergraphs whose minimum distance improves 
earlier estimates of the distance of bipartite-graph codes. 
They also suggested a decoding algorithm for such codes and estimated 
its error-correcting capability.

In this paper we study two aspects of hypergraph codes. First, we compute the
weight enumerators of several ensembles of such codes, establishing conditions
under which they attain the Gilbert-Varshamov bound and deriving estimates
of their distance. In particular, we show that this bound is attained by codes
constructed on a fixed bipartite graph with a large spectral gap.

We also suggest a new decoding algorithm of hypergraph codes that corrects a 
constant fraction of errors, improving upon the algorithm of Bilu and Hoory. 
\end{abstract}
\vskip-1cm 
\maketitle



\section{Introduction}
Codes on graphs account for some of the best known code families in terms 
of their error correction under low-complexity decoding algorithms. 
They are also known to achieve a very good tradeoff between the 
rate and relative distance.
The most well-studied case is codes defined on a bipartite graph.
In this construction, a code of length $N=mn$ is obtained by ``parallel
concatenation'' of $2m$ codes of a small length $n$ which refers to
the fact that each bit of the codeword is checked by two independent length-$n$
codes. The arrangement of parity checks is specified by the edges of
a bipartite graph which are in one-to-one correspondence with the codeword bits.

Codes on bipartite graphs are known to be asymptotically
good, i.e., to have nonvanishing rate $R$ and relative distance $\delta$ as 
the code length $N$ tends to infinity. Constructive families of bipartite-graph
codes with the best known tradeoff between $R$ and $\delta$ have been
found by the present authors \cite{bar06}. In particular, codes constructed in that paper
surpass the product bound on the minimum distance which is a common
performance benchmark for concatenated constructions.

Moving from constructive families to existence results obtained by averaging
over ensembles of bipartite-graph codes, it is possible to derive even better
rate-distance tradeoffs. In particular, bipartite-graph codes with random
local codes and random bipartite graphs 
attain the Gilbert-Varshamov (GV) bound for relatively small
code rates and are only slightly below it for higher rates \cite{bar06}.

A natural way to generalize codes on bipartite graphs is to consider
concatenations governed by regular $t$-partite hypergraphs, $t\ge 2.$
This code family was studied by Bilu and Hoory 
in \cite{bil04}. While constructive families
of bipartite-graph codes rely on the expansion property of the underlying 
graph, expansion is not well defined for hypergraphs.
Instead, \cite{bil04} put forward a property of hypergraphs, called
$\epsilon$-homogeneity, which replaces expansion 
in the analysis of hypergraph codes.
\cite{bil04} showed that there exist explicit, easily constructible
families of $\epsilon$-homogeneous hypergraphs, and estimated the number of
errors corrected by their codes under a decoding algorithm suggested in that
paper.

In this paper we study hypergraph codes both from the perspective
of weight distributions and their decoding. The results of \cite{bar06}
on weight distributions are advanced in several directions.
In Theorem \ref{thm:random} and its corollary
we prove that the code ensemble defined by random regular $t$-partite
hypergraphs and random local linear codes contains codes that meet the 
GV bound. The region of code rates for which this claim holds true
extends as $t$ increases from the value $t=2.$ 
We also show (Theorem \ref{thm:fixed}, Cor.~\ref{cor:fixed})
that the ensemble of hypergraph codes contains codes that attain the 
GV bound even if random hypergraphs are replaced with a
{\em fixed} $\epsilon$-homogeneous hypergraph. Specializing the last result
for $t=2$, we establish that expander codes
of Sipser and Spielman \cite{sip96} constructed from a fixed graph with a
large spectral gap and random local codes with high probability attain
the GV bound.
Finally, we derive an estimate of the average weight distribution for the
ensemble of hypergraph codes with a fixed local code (see Theorem \ref{thm:C2})
that refines
substantially a corresponding result in \cite{bar06} and generalizes it
from $t=2$ to arbitrary $t$.

The tradeoff between the rate and relative distance of hypergraph codes
shows an improvement over bipartite-graph codes for small values of the 
distance. On the other hand, the decoding algorithm of \cite{bil04}
does not exploit the full power of their codes; moreover, for 
small $\delta$ the proportion of errors corrected by it vanishes compared
to the value of the distance. Motivated by this, 
we propose a new decoding algorithm of hypergraph codes 
and estimate its error-correcting capability. 
We show that it corrects the number of
errors which constitutes a fixed proportion of the code's distance.

\subsection{Codes on bipartite graphs}
Let $G=(V,E)$ be a balanced, $n$-regular bipartite graph with the
vertex set $V=V_1\cup V_2, |V_1|=|V_2|=m$ and $|E|=N=nm$ edges.
Let us choose an arbitrary ordering of edges in $E$. For a given vertex 
$v\in V$ this defines an ordering of edges $v(1),v(2),\dots,v(n)$
incident to it. We denote this subset of edges by $E(v)$. 
Given a binary vector $x\in \{0,1\}^N,$ let us establish a one-to-one
correspondence between the coordinates of $x$ and the edges in $E$.
For a given vertex $v$ let $x(v)=(x_e, e\in E(v))$ be the subvector that 
corresponds to the edges in $E(v).$ Denote by $\lambda$ the second largest
in the absolute value eigenvalue of the graph $G$.

Consider a set of binary linear codes $A_v[n, R_1n]$ of length $n$
and rate $R_1\triangleq \dim(A_v)/n,$ where $v\in V.$ 
Define a {\em bipartite-graph code}
as follows:
  $$
   C(G,\{A_v\})=\{x\in \{0,1\}^N: \forall_{v\in V_1\cup V_2} x(v)\in A_v\}.
  $$
The rate of the code $C$ is easily seen to satisfy
  \begin{equation}\label{eq:rate2}
  R(C)\ge 2R_1-1.
 \end{equation}
If we assume that all the local codes are the same, i.e., $A_v=A,$ where
$A[n,R_1n,d_1=\delta_1 n]$ is some linear code, then the
distance of the code $C$ can be estimated as follows:
  $$
   d/N\ge \delta_1^2 \Big(1-\frac\lambda{d_1}\Big)^2
  $$
(we will write $C(G,A)$ instead of $C(G,\{A\})$ in this case).
In particular, if the spectral gap of $G$ is large, i.e., 
$\lambda$ is small compared to $d_1$, then the relative
distance $d/N$ is close to the value $\delta_1^2$, similarly to the case of the
direct product 
code $C=A\otimes A$.

The weight distribution of bipartite-graph codes constructed from random 
regular bipartite graphs and a fixed local code $A$ with a known weight
distribution was analyzed in \cite{bou99,len99}. In particular, it was shown
that if $A$ is the Hamming code $\cH_m$ then the ensemble $\cC=(C(G,A))$
contains asymptotically good codes. Generalizing these results,
paper \cite{bar06} studied the weight distribution of 
bipartite-graph codes
with fixed and random component codes $A$. It was shown that for $m\to\infty$
the 
ensemble of codes constructed from random regular 
bipartite graphs and
a fixed code $A$ with distance $d_1\ge 3$ contains asymptotically good
codes. 
It has also been shown \cite{bar06} that if the local codes are chosen
randomly, 
then the code ensemble $\cC$ 
contains codes that meet the GV bound in the interval
of code rates $R(C)\le 0.2.$ 

\subsection{Codes on hypergraphs}
Generalizing the above construction, 
let $H=(V,E)$ be a $t$-uniform $t$-partite $n$-regular hypergraph.
This means that the set of vertices $V=V_1\cup\dots\cup V_t$ of $H$
consists of $t$ 
disjoint parts of equal size, say, $|V_i|=m, 1\le i\le t.$ 
Every hyperedge $\{v_{i_1},v_{i_2},\dots,v_{i_t}\}$ contains exactly $t$
vertices, one from each part, and each  vertex is incident to $n$ hyperedges.
Below for brevity we say edges instead of hyperedges.
The number of edges of $H$ equals $N=mn$ which will also be the length
of our hypergraph codes.
As above, assume
that the edges are ordered in an arbitrary fixed way and denote by
$E(v)$ the set of edges incident to a vertex $v$.  
For definiteness, let us assume that edges $e_{(i-1)n+j},j=1,\dots,n$
are incident to the vertex $v_i\in V_1, 1\le i\le m.$
  
Given a binary vector $x\in\{0,1\}^N$ whose coordinates are in a one-to-one
correspondence with the edges of $H$ denote by $x(v)$ its subvector 
that corresponds to the edges in $E(v).$ 

Define a 
{\em hypergraph code} as follows: 
   $$
  C(H,\{A_v\})=\{x\in\{0,1\}^N: \forall_{v\in V} x(v)\in A_v\},
   $$
where $\{A_v, v\in V\}$ is a set of binary linear codes of length $n$. 
As above, if all the codes are the same, we write $C(H,A).$
Assume that all the codes $A_v$ have the same rate $R_1$, then
the rate of the code $C$ satisfies
  \begin{equation}\label{eq:rate-t}
    R(C)\ge tR_1-(t-1).
  \end{equation}
\begin{definition} \cite{bil04}
A hypergraph $H$ is called $\epsilon$-homogeneous if
for every $t$ sets $D_1,D_2,\dots,D_t$ with $D_i\subseteq V_i$ and
$|D_i|=\alpha_i m,$
  \begin{equation}
  \label{eq:homog}
   \frac{|E(D_1,D_2,\dots,D_t)|}{N}\le \prod_{i=1}^t \alpha_i
        +\epsilon\min_{1\le i<j\le t}\sqrt{\alpha_i\alpha_j},
\end{equation} 
where $E(D_1,D_2,\dots,D_t)$ denotes the set of edges that intersect
all the sets $D_i.$ 
\end{definition}
This definition quantifies the deviation of the hypergraph $H$ from
the expected behavior of a random hypergraph.
For $t=2$ the well-known ``expander mixing lemma''
asserts that
   $$
   \Big|\frac{|E(D_1,D_2)|}{N}-\alpha_1\alpha_2\Big|
         \le \frac \lambda n \alpha_1\alpha_2,
   $$
showing that regular bipartite graphs are $\lambda/n$-homogeneous.
This inequality is frequently used in the analysis of bipartite-graph codes
\cite{sip96,zem01}.

Let $A[n,R_1n,d_1=\delta_1 n]$ be a binary linear code.
The distance of a code $C(H,A)$ where $H$ is $\epsilon$-homogeneous satisfies
\cite{bil04}
  \begin{equation}\label{eq:dist-t}
    d/N\ge \delta_1^{\frac t{t-1}}-c_1(\epsilon,\delta_1,t)
  \end{equation}
where $c_1\to 0$ as $\epsilon\to 0.$

One of the main results in \cite{bil04} gives an explicit construction
of $\epsilon$-homogeneous hypergraphs $H$ starting with a regular graph 
$G(U,E)$ with degree $\Delta$ and second eigenvalue $\lambda.$ 
Putting $V_i=U, i=1,2,\dots,t$ and introducing a hyperedge whenever the $t$
vertices in the graph $G$ are connected by a path of length $t-1$, that paper
shows that the resulting hypergraph is $n$-regular and
$\epsilon$-homogeneous with $n=\Delta^{t-1}, \epsilon=2(t-1)\lambda/\Delta.$
Therefore, starting with a family of $\Delta$-regular bipartite graphs with
a large spectral gap, one can construct a family of regular
homogeneous hypergraphs with a small value of $\epsilon.$
Paper \cite{bil04} has also established that random $n$-regular
hypergraphs with high probability are $O(1/\sqrt n)$-homogeneous.

\section{Weight distributions}
Below we consider ensembles of random codes on graphs and hypergraphs.
In some cases the (hyper)graph will be selected randomly. In the case
of bipartite graphs this is done as follows. Connect the edges $e_{(i-1)n+j}, j=1,\dots,n$ to the 
vertex $v_i\in V_1,$ $i=1,\dots,m$. Next choose a permutation on the 
set $E$ with a uniform distribution and connect the remaining half-edges 
to the vertices in $V_2$ using this permutation. Similarly, to construct 
an ensemble of random hypergraphs,
we choose $t-1$ permutations independently with uniform distribution
and use them to connect the parts of $H$.

Random linear codes are selected from the standard ensemble of length-$n$
codes defined by $n(1-R_1)\times n$ random binary matrices whose entries
are chosen independently with a uniform distribution. 

We consider the following three ensembles of hypergraph codes.

\medskip 
{\bfsl Ensemble} $\cC_1(t).$ A code $C(H,\{A_1,\dots,A_t\})\in \cC_1(t)$ is constructed 
by choosing a random $t$-partite hypergraph $H$ and choosing
random local linear 
codes $A_i$ of length $n$ independently 
for each part $V_i\in V.$ 

\smallskip
{\bfsl Ensemble} $\cC_2(t,A).$ A code $C(H,A)\in \cC_2$ is 
constructed by choosing a random $t$-partite hypergraph $H$
and using the same fixed local code $A[n,R_1n,d_1]$ as a local 
code at every vertex.

\smallskip
{\bfsl Ensemble} $\cC_3(t,H).$ A code $C(H,\{A_v\})$ from this ensemble
is formed by choosing a fixed, nonrandom hypergraph $H$ and 
taking  random local linear codes $A_v$ independently 
for each vertex $v\in V$.

\medskip 
Our purpose is to compute ensemble-average asymptotic weight distributions 
for codes in these ensembles and to estimate the average minimum distance
assuming that $m\to\infty$ and $n$ is a constant. The case $t=2$ 
corresponds to ensembles of bipartite-graph codes, some of
which were studied in \cite{bar06,bou99,len99}. 
Below we will cover the remaining cases for the code ensembles $\cC_i(t),$
$i=1,2,3$ and any $t\ge 2$. 
Below $B_w=B_w(C)$ 
denotes the number of codewords of weight $w$.
Before proceeding, we note that upper bounds on the ensemble-average weight distribution 
in many cases also give a lower bound on the code's distance.
\begin{lemma} 
Suppose that for an ensemble of codes $\cC$ of length $N$
there exists an $\omega_0>0$ such that 
       $$\lim_{N\to\infty} \sum_{w\le\omega_0 N} \ee B_w=0.$$
Then for large $N$ the ensemble contains codes whose relative distance
satisfies $d/N\ge \omega_0.$ 
\end{lemma}
The proof is almost obvious because
    $$
      \Pr[d(C)\le\omega_0 N]\le \sum_{w\le \omega_0 N} \Pr[B_w(C)\ge 1]
        \le \sum_{w\le \omega_0 N}\ee B_w(\cC).
    $$
\medskip
\subsection{Ensemble $\cC_1(t)$}
\begin{theorem} \label{thm:random} 
For $m\to \infty$ the average weight distribution 
over the ensemble of linear codes $\cC_1(t)$ of length $N=mn$ and rate
(\ref{eq:rate-t}) satisfies 
$\ee B_{\omega N}\le 2^{N(F+\gamma)},$ where 
 \begin{equation}\label{eq:b}
  F=\omega t\log_2(2^{(1-R)/t}-1)-(t-1)h(\omega), 
         \quad\text{ if\;\;  } 0\le\omega\le 1-2^{(R-1)/t}
 \end{equation}
 \begin{equation}\label{eq:a}
 F=h(\omega)+R-1 \quad\text{ if \;\; }\omega\ge 1-2^{(R-1)/t},
 \end{equation}
and 
  $$
  \gamma\le tn^{-1}(1+\log_2n)+(t/2N)\log_2(2N),
  $$
$h(z)=-z\log_2 z-(1-z)\log_2(1-z).$\end{theorem}
\begin{proof} 
The proof is an extension of the corresponding result for $t=2$ in 
\cite{bar06}.
Let $C_i, i=1,\dots, t$ be the set of vectors 
$x\in \{0,1\}^N$ that satisfy the linear constraints of part $V_i$
of the hypergraph $H$ so that $C(H,A)=\cap_i C_i.$ Let $P_i=\Pr[x\in C_i].$
The events $x\in C_i$ for different $i$ are independent, and therefore
  $$
    \Pr[x\in C]=P_i^t
  $$
(for any $i=1,\dots,t).$ Let $B_w(C_i)$ be the random number of vectors of 
weight $w$ in the code $C_i$. Then
  $$
   \ee B_w(C)=\binom Nw \Pr[x\in C]= \binom Nw \prod_{i=1}^t \frac {\ee B_w(C_i)}{\binom Nw}.
  $$
Let $\cX_{s,w}$ be the set of vectors of weight $w=\omega N$ 
whose nonzero coordinates are incident to some vertices 
$v_{i_1},\dots,v_{i_s}\in V_1,\, s\ge w/n.$ 
Let $w_j=\wt(x(v_{i_j})),j=1,\dots,s$ and let $\omega_j=w_j/n.$ We have
  $$
  |\cX_{s,w}|=\binom ms \sum_{\begin{substack}{w_1,\dots,w_s\\
     \sum w_j=w}\end{substack}} \prod_{j=1}^s \binom n{w_j}
   \le \binom ms\sum_{\begin{substack}{w_1,\dots,w_s\\
     \sum w_j=w}\end{substack}} 2^{n\sum_j h(\omega_j)}.
  $$
By convexity of the entropy function, the maximum of the last expression
on $\omega_1,\dots,\omega_s$ under the constraint $n\sum_j\omega_j=\omega N$
is attained for $\omega_j=\omega m/s, j=1,\dots,s.$
Since the sum contains no more than $n^s$ terms, we obtain
  $$
  |\cX_{s,w}|\le 2^{m h(x)+ s\log n+s nh(\omega m/s)}
    \le 2^{N(x h(\omega/x)+\epsilon)}
  $$
where $x=s/m$ and $\epsilon=(1+\log n)/n.$ A vector $x\in \cX_{s,w}$
is contained in $C_1$ with probability $2^{sn(R_1-1)}$. Thus,
  $$
   \ee B_w(C_1)=|\cX_{s,w}|2^{sn(R_1-1)},
  $$
and the same expression is true for $\ee B_w(C_i), i=2,\dots,t.$
Therefore,
  $$
  \ee B_w(C)\le \binom Nw^{-(t-1)}2
^{tN (\max_{\omega\le x\le 1}(x(h(\omega/x)+R_1-1))
    +\epsilon)}.
  $$ 
Since $t(R_1-1)\le R-1,$ we obtain $\ee B_w(C)\le 2^{N(F(\omega)+\gamma)},$
where
  $$
    F(\omega)\le -(t-1)h(\omega)+t\max_{\omega\le x\le 1}(x(R_1-1+h(\omega/x)))
  $$
  $$
\le -(t-1)h(\omega)+\max_{\omega\le x\le 1}(x(R-1+th(\omega/x))).
  $$
 The maximum on $x$ of $x(R-1+th(\omega/x))$ is attained for $x=x_0=
\omega/(1-z)$ where $t\log_2 z=R-1$. The two cases in the theorem are obtained
depending on whether $x_0<1$ or not. If $x_0<1$, we substitute $x_0$ in the
expression for $F(\omega)$ and obtain
  $$
   F(\omega)\le -(t-1)h(\omega)+\omega t\log_2\frac{z}{1-z}
  $$
which implies (\ref{eq:b}) on account of the identity $R-1+th(z)=t(1-z)
\log_2(z/(1-z))$. If $x_0\ge 1$, we substitute
 the value $x=1$ to obtain (\ref{eq:a}). 
  \end{proof}

\begin{corollary} Let $\omega^\ast$ be the only nonzero root of the equation
  $$
   \omega\Big(R-1-t\log_2\Big(1-2^{(R-1)/t}\Big)\Big)=(t-1)h(\omega).
  $$
Then the average relative distance over ensemble $\cC_1(t)$ behaves as
  \begin{align*}
   \delta(R)&\ge \omega^\ast,  &\text{if\;\;} R\le \log_2(2(1-\dgv(R))^t)\\
   \delta(R)&\ge \dgv(R),      &\text{if\;\;}R>\log_2(2(1-\dgv(R))^t),
  \end{align*}
where $\dgv(x)\triangleq h^{-1}(1-x).$
\end{corollary}
The proof is analogous to the proof of Corollary 4 in \cite{bar06} and will
be omitted.

For $t=2$ we proved in \cite{bar06} that ensemble $\cC_1$ contains
codes that reach the GV bound if the code rate satisfies
$0\le R\le 0.202.$ This result forms a particular case of the above 
corollary. Increasing $t$, we find that the ensemble contains codes that 
reach the GV bound for the values of the rate as shown below:
  \begin{alignat*}{3}
     &t=3       & &\quad4      & &\quad10\\
      &R\le 0.507& &\quad0.737& &\quad0.998     .
  \end{alignat*}
Thus already for $t=10$ almost all codes in the ensemble $\cC_1$ 
attain the GV bound for all but very high rates.

\subsection{Ensemble $\cC_2(t,A)$.} 
In this case the results depend on the amount of 
information available for the local codes. Specifically, \cite{bar06} shows
that for $t=2$ the ensemble contains asymptotically good codes provided
that the distance of the local code $A$ is at least 3.
In the case when the weight distribution of the code $A$ is known,
a better estimate is known from \cite{bou99,len99}.
\begin{theorem}\label{thm:chernov} Let $A$ be a linear code of
length $n$ with weight enumerator $a(x)=\sum_{i=0}^n a_ix^i.$ 
Let $B_w$ the random number of codewords of weight $w$ of a code 
$C(H,A)\in \cC_2(t,A).$ Then its average value
over the ensemble satisfies
\begin{align*}
   \lim_{N\to\infty}\frac 1N \log_2 \ee B_{\omega N}\le 
       &-(t-1)h(\omega)+
       \frac t{\ln2}\Big(\frac1n\ln a(e^{s^\ast})-s^\ast\omega\Big),
\end{align*}
where $s^\ast$ is the root of $\;(\ln a(e^s))'_s=n\omega.$ 
\end{theorem}
This theorem enables us to estimate the 
asymptotics of the mean relative distance
$\delta=\lim\limits_{m\to\infty} \frac{\ee d(C)}{N}$ for the ensemble 
$\cC_2$. Let us consider several examples.

\medskip
1. Let $t=3$ and let $A$ be the Hamming code of length $n=15$
and rate $R_1=11/15.$ Then the rate $R(\cC_2)\ge 0.2$ and the distance 
$\delta=0.2307$. The relative GV distance for this rate is
$\dgv(0.2)=0.2430.$

\smallskip
2. Let $t=3$ and let $A$ be the Hamming code of length $n=31$. Then
$R(\cC_2)\ge 16/31$ and $\delta=0.0798.$ Using the same code with $t=4$
gives $R(\cC_2)\ge 11/31$ and $\delta=0.1607$ while $\dgv(11/31)=0.1646.$

\smallskip
3. Let $t=3$ and let $A$ be the 2-error-correcting primitive
BCH code of length $n=31$ and rate $R_1=21/31.$ Then 
the rate $R(\cC_2)\ge 1/31$ and the value of $\delta$ is 
$\approx 0.3946608$. The relative GV distance 
for this rate is $\dgv(1/31)=0.3946614.$

\medskip
\remove{
\begin{theorem} (a) Let $A$ be a $[n,R_0n,d\ge t/(t-1)]$ linear code.
Then there exists an $\omega_0>0$ such that 
  $$
       \lim_{m\to\infty}\sum_{w\le \omega_0 N} \ee B_{w}=0.
  $$
In particular, the ensemble $\cC_2(t,A)$ contains asymptotically good codes.
\end{theorem} 
}
Let us turn to the case when only the minimum distance $d_1$ of the code
$A$ is available. In \cite{bar06} we addressed the case $t=2$, proving 
that as long as $d_1\ge 3$, there exists an $\epsilon>0$ such that the 
ensemble-average relative distance $\delta>\epsilon$ as $m\to\infty.$
In the next theorem this result is extended to arbitrary $t\ge 2.$
We also prove a related result which gives an upper
bound on the average weight spectrum and provides a way of estimating the
value of $\omega_0.$ 

\begin{theorem}\label{thm:C2} (a) Let $A$ be the local code of length $n$ 
and distance $d_1$
used to construct
the ensemble $\cC_2(t,A)$ of hypergraph codes. Let $x_0=x_0(\omega)$ be the positive 
solution of the equation
   \begin{equation}\label{eq:root1}
    \omega n+\sum_{i=d_1}^n\binom ni(\omega n-i)x^i=0.
   \end{equation}
The ensemble-average weight distribution satisfies
  $$
    \lim_{N\to\infty} \frac 1N \log \ee B_{\omega N}\le \frac tn
      \log\frac{1+\sum_{i=d_1}^n\binom ni x_0^i}{x_0^{\omega n}}-(t-1)h(\omega).
  $$
(b) The inequality $d_1> t/(t-1)$ gives a sufficient condition for the ensemble 
to contain asymptotically good codes.
\end{theorem}
\begin{proof} In the proof we write $d$ instead of $d_1$ to refer to the 
distance of the code $A$.

(a) Let $H$ be a random hypergraph and $C(H,A)$ be the corresponding
code. Recall that $C=\cap_i C_i,$ where $C_i$ is the set of vectors
that satisfy the constraints of part $i$ of the graph. 
Let $U_i(w,d)$ be the set of vectors $x\in \{0,1\}^N$
such that $\wt(x)=w$ and $\wt(x(v))=0$ or $\wt(x(v))\ge d$ 
for all $v\in V_i.$ Since the number of such vectors is the same for all $i$,
below we write $|U(w,d)|$ omitting the subscript.
Let us choose a vector $x\in\{0,1\}^N$ randomly with a uniform
distribution. Then
  $$
   \Pr[x\in C_1|\wt(x)=w]\le \frac{|U(w,d)|}{\binom Nw}
  $$
and for $i\ge 2,$
  $$
   \Pr[x\in C_i|\wt(x)=w,x\in C_1]=\Pr[x\in C_i|\wt(x)=w].
  $$
Then
  $$
    \ee B_w(C)=\binom N w\Pr[x\in C|\wt(x)=w]=\binom Nw 
   (\Pr[x\in C_1|\wt(x)=w])^t
  $$
  \begin{equation}\label{eq:U}
    \le \frac{|U(w,d)|^t}{\binom Nw^{t-1}}
  \end{equation}
Given a vector $x$ denote by $j_\ell$ the number of vertices $v\in V_i$
such that $\wt(x(v))=\ell.$ Clearly, 
   $$
     |U(w,d)|= \sum_{\begin{substack}
        {j_0,j_d,j_{d+1},\dots,j_n\\ \sum \ell j_\ell=w,\,j_0+\sum
\limits_{\ell\ge d}j_{\ell}=m}
       \end{substack}}
      \binom m{j_0,j_d,\dots,j_n}\prod_{\ell=d}^n \binom n \ell^{j_\ell}
   $$
This sum contains no more than $(m+1)^n=O(N^n)$ terms, 
so for $N\to\infty$ its exponent is determined by the maximum term (which
has exponential growth).
\remove{Taking the logarithm and letting $w=\omega N$, we obtain
  $$
   \log |U(w,d)|^t=t\log \Bigg\{\sum_{\begin{substack}
        {j_d,j_{d+1},\dots,j_n\\ \sum \ell j_\ell=w,\,\sum j_\ell=m}
       \end{substack}}
      \binom m{j_d,\dots,j_n}\prod_{\ell=d}^n \binom n \ell^{j_\ell}\Bigg\}
  $$
}
We obtain
 \begin{equation}\label{eq:h}
   \frac 1N\log|U(\omega N,d)|^t \le \frac t n 
      \max_{\begin{substack}{\nu_0,\nu_d,\dots,\nu_n\\[1pt] 
         \sum \ell\nu_\ell=\omega n,\,\sum \nu_\ell=1 }\end{substack}}
     \Big\{h(\nu_0,\nu_d,\nu_{d+1},\dots,\nu_n)
     +\sum_{\ell=d}^n \nu_\ell\log\binom n\ell\Big\}+\frac{\log N}m,
  \end{equation}
where $\nu_\ell=j_\ell/m, \ell=0,d,d+1,\dots,n,$ and $h(\underline x)=-
\sum_{i}x_i\log x_i.$
The objective function is concave, so the point of extremum is found
from the system of equations
  \begin{align*}
    \binom ni( {1-\sum_{\ell=d}^n \nu_\ell})&={\nu_i}\mu^{-i}, \quad i=d,d+1,\dots,n\\
    \sum_{\ell=d}^n \ell\nu_\ell&=\omega n.
  \end{align*}
Its solution is given by
 $$
    \nu_i=\frac{\binom ni\mu^{i}}{1+
          \sum_{\ell=d}^n\binom n \ell\mu^{\ell}}
            , \quad i=d,d+1,\dots,n.
  $$         
where $\mu$ is chosen so as to satisfy the last equation of the system.
Evaluating $\sum_i i\nu_i$ and writing $x$ instead of $\mu,$ we observe
that it should satisfy Eq. (\ref{eq:root1}). 
This equation has a unique root $x_0>0$ 
because putting $x=p/(1-p),$ we can write it as
  $$
    \omega n\Big(\frac{\Pr[X=0]}{\Pr[X\ge d]}+1\Big)=\ee[X|X\ge d],
  $$
where $X$ is a binomial $(p,1-p)$ random variable. As $p$ changes from
0 to 1, the left-hand side of the last equation decreases monotonically
from $+\infty$ to $\omega n$ while the right-hand side increases 
monotonically from $d$ to $n$.

Finally, computing the entropy and simplifying, we obtain the estimate
  $$
   \lim_{N\to\infty}\frac 1N\log|U(\omega N,d)|^t \le
      \log\frac{1+\sum_{i=d}^n\binom ni x_0^i}{x_0^{\omega n}}.
  $$

\vspace{3ex}
(b) The proof of the second part is analogous to the case of $t=2$ in \cite{bar06}.
Let $w,1\le w\le N$ be the weight and let $p=w/d.$
We have
   \begin{align*}
    |U(w,d)|&\le \sum_{i=w/n}^p \binom ni\binom n d^i\binom{in}{(p-i)d}
  \\
  &   \le \binom m p\binom nd^p\sum_{i=w/n}^p \binom {pn}{(p-i)d}\\
  &
   \le \binom mp \binom nd^p2^{pn}.
  \end{align*}
Then 
  $$
    \ee B_w(C)\le \Big(\binom mp \binom nd^p2^{pn}\Big)^t\binom Nw^{1-t}.
  $$
Using the estimates  $(\frac nk)^k\le\binom nk\le (\frac{en}k)^k$, we compute
  \begin{align*}
   \ee B_w(C)&\le 
   \Big(\frac {em}p\Big)^{pt}n^{dpt}2^{tpn}\Big(\frac w N\Big)^{w(t-1)}\\
    &=(sm/w)^{\frac wd(t-d(t-1))}
  \end{align*}
where $s=((ed2^n)^t n^d)^{\frac1{t-d(t-1)}}.$ Thus, for any $\omega$ satisfying
$\omega<s/m,$ the average number of vectors of weight $\omega N$ tends to 0 as
$m\to\infty$ as long as $d (t-1)>t.$ This proves that under this condition the
ensemble contains asymptotically good codes.
\end{proof}

{\em Examples.} Let $A$ be the $[7,4,3]$ Hamming code and let $t=2$.
Theorem \ref{thm:C2}(a) implies a lower bound $\delta\ge 0.01024$
on the average relative distance for the ensemble $\cC_2(2,A).$
This improves upon previous results (\cite{bou99,len99}; also Part (b) of this
theorem) which assert only that the ensemble contains asymptotically good
codes.
Of course, in this case we can use the entire weight distribution of the
code $A$ to find the estimate $\delta\ge 0.186$ from Theorem \ref{thm:chernov};
however, in cases when the weight distribution is difficult to find, 
the last theorem provides new information for the ensemble of graph codes.

Similarly, for $A[23,12,7]$
from Theorem \ref{thm:C2}(a) we obtain the estimate $\delta\ge 0.0234$.
Again, using the entire weight distribution, it is possible to obtain a
better estimate.

\bigskip Part (a) of the last theorem implies the following corollary
which shows what happens to the average weight spectrum of the ensemble
for long local codes.
\begin{corollary} Let $d_1=\delta_1 n.$ Then
    $$
     \frac 1N \log \ee B_{\omega N}(C)\le
         \frac{t\omega}{\delta_1} h(\delta_1)-(t-1)h(\omega)+\gamma
   $$
where $\gamma\le (\log N)/m+(\log n)/n.$
\end{corollary}
\begin{proof}
In (\ref{eq:h}) let us bound above $h(\cdot)$ by $\log n.$
Then
   $$
    \frac 1N\log|U(w,d_1)|^t \le
     \frac tn \max_{\begin{substack}{\nu_{d_1},\dots,\nu_n\\[1pt] 
         \sum \ell\nu_\ell=\omega n, }\end{substack}}
     \sum_{\ell=d_1}^n \nu_\ell\log\binom n\ell+\gamma.
  $$ 
Computing the maximum amounts to solving a linear programming problem 
whose dual is 
$$\omega nz\to\min $$
$$\ell z\ge \log \binom n\ell, \ell=d_1,d_1+1,\dots, n;\; z\ge 0.
$$
Its solution is given by 
   $z^\ast=\omega n\max_{d_1\le \ell\le n} \log\binom n\ell/\ell.$
We obtain
   $$
     \frac 1N\log|U(w,d_1)|^t \le t  \omega 
     \max_{\delta_1 \le x\le 1}\frac {h(x)}x+\gamma
           \le t\omega h(\delta_1)/\delta_1+\gamma.
     $$
Employing (\ref{eq:U}) now completes the proof. 
\end{proof}

\remove{2. Let us provide a few more examples.\\

\begin{tabular}{llll} local code & $t=2$ & $t=3$ & $t=4$\\[3pt]
  $A[31,26,3]$ & &0.01027&0.01026
  \end{tabular}}

\subsection
{ Ensemble $\cC_3(t,H)$} 

\begin{theorem} \label{thm:fixed} Assume that $H$ is $\epsilon$-homogeneous.
For $m\to\infty$ the average weight distribution over the
ensemble of linear codes $\cC_3(t,H)$ satisfies $\ee B_{\omega N}\le 
2^{N(F+\gamma)}$ where
   $$
    F=-x_0(1-R)+x_0^t h\Big(\frac\omega{x_0^t}\Big), \quad\text{if }
    x_0< 1
  $$
  $$
    F=h(\omega)+R-1, \quad\text{if }x_0\ge 1,
  $$
where $x_0$ is the unique positive root of the equation 
  \begin{equation}\label{eq:root2}
  tx^{t-1}\log(x^t/(x^t-\omega))=1-R
  \end{equation}
$\gamma=t(n+\log m)/N+\epsilon.$
\end{theorem}
\begin{proof} Let $C\in \cC_3(t,H)$ and let 
$x\in\{0,1\}^N$ be a nonzero vector. 
Denote by $B_i$ the set of nonzero vertices of $x$ in the part
$V_i, i=1,\dots,t.$ Let $E=|E(B_1,B_2,\dots,B_t)|.$ Let 
$b_i=|B_i|, \beta_i=b_i/m,$ then the probability that $x\in C$ 
equals $2^{-(1-R_1)N\sum_i \beta_i}.$
Assume w.l.o.g. that $\beta_1<\beta_2<\dots<\beta_t.$ The average number
of vectors of weight $w=\omega N$ in the code $C$ can be bounded above as
  $$
    \ee B_{w}\le \sum_{\omega m\le b_1,b_2,\dots, b_t\le m}
   \binom{\prod_{i=1}^t b_i +\epsilon\sqrt{b_1b_2}}{w}
   \prod_{i=1}^t\binom n{b_i}2^{-(1-R_1)N\sum_i \beta_i}.
  $$
Then
  $$
   \frac 1N\log \ee B_{\omega N}\le \max_{\begin{substack}{\omega\le \beta_i\le 1\\
    \prod_i\beta_i\ge \omega}\end{substack}}
   \Big\{\prod_{i} \beta_i h\Big(\frac\omega{\prod_i
   \beta_i}\Big)-(1-R_1)\sum_i \beta_i\Big\}+\gamma
  $$
Let $\phi(\beta_1,\dots,\beta_t)$ be the function in the brackets in the
last expression.
%
Let us prove that $\phi$ is concave in the domain
$\cD=\prod_i [\omega,1]\cap\{(\beta_1,\dots,\beta_t):\prod_i\beta_i\ge \omega\}$. Computing its Hessian matrix, we obtain
   $$
     H_\phi=-\log e\left[\begin{matrix} \frac{s_1}{\beta_1^2}
         &\frac{s_2}{\beta_1\beta_2}&\dots&\frac{s_2}{\beta_1\beta_t}\\
     \frac{s_2}{\beta_2\beta_1}&\frac{s_1}{\beta_2^2}&\dots
      &\frac{s_2}{\beta_2\beta_t}\\
     \vdots&\vdots&\ddots&\vdots\\
   \frac{s_2}{\beta_t\beta_1}&\frac{s_2}{\beta_t\beta_2}&\dots
           &\frac{s_1}{\beta_t^2}
   \end{matrix}\right]
  $$
where
  $$
    s_1=\frac{\omega \prod_{i} \beta_i  }{\prod_i \beta_i-\omega}
  $$
  $$
   s_2=s_1    +\prod_i\beta_i\ln\Big(1-\frac\omega{\prod_i\beta_i}\Big).
  $$
The matrix $H_\phi$ can be written as
  $$
 H_\phi=-\log e(s_2 z z^t+(s_1-s_2)
\text{diag}(\beta_1^{-2},\dots,\beta_t^{-2}))
  $$
where $z=(1/\beta_1,\dots,1/\beta_t)^t$ and diag$(\cdot)$ denotes
a diagonal matrix. We wish to prove that $H_\phi$ is negative definite
for $\beta_i>0, 0<\omega<\prod_i \beta_i.$
Clearly, $s_1> s_2,$ and therefore the claim will follow if we show that 
$s_2>0.$
This is indeed true because letting $Q=\prod_i \beta_i$ and using
the inequality $x>\ln(1+x)$ valid for $x>-1, x\ne 0,$ we have
  $$
    s_2=Q\Big(\frac\omega{Q-\omega}
      +\ln\frac{Q-\omega}{Q}\Big)> Q\Big(\ln\Big(1+\frac\omega{Q-\omega}\Big)
          +\ln\frac{Q-\omega}{Q}\Big)=0.
  $$

We will now show that the maximum of $\phi$ in $\cD$ is attained on the line $\ell$ given by 
$\beta_1=\beta_2=\dots=\beta_t.$
Note that $\cD$ is an 
intersection of convex domains and therefore itself convex.
\remove{
To show that the region 
$\{(\beta_1,\dots,\beta_t):\prod_i\beta_i\ge \omega\}$ is convex, we again compute the Hessian matrix of the function
$\omega/\prod_i\beta_i$ and observe that it is a constant multiple
of a sum of two matrices, one of which is a diagonal matrix with positive
entries and the other a rank-one matrix.}
Moreover, the domain $\cD$ is also symmetric in the sense that together with
any point $p=(\beta_1,\dots,\beta_t)$ it also contains all the points
obtained from $p$ by permuting its coordinates, and the value of $\phi$
at each of these points is the same and equal to $\phi(p).$ 
Because $\phi$ is strictly concave, for any point $p\in\cD,p\not\in\ell$ 
it is possible to find a point $q$ such that 
$\phi(q)>\phi(p)$ (any point $q$ on the segment between $p$
and one of its symmetric points will do). 
This shows that the global maximum of $\phi$ in $\cD$
is attained on $\ell$ including possibly the point 
$\beta_1=\dots=\beta_t=1.$ 
Thus, we obtain
  $$
   \frac 1N\log \ee B_{w}\le \max_{\omega^{1/t}\le x\le 1}
   \{-(1-R)x +x^t h\Big(\frac\omega{x^t}\Big)\}+\gamma.
  $$ 
The maximum of this expression on $x$ is attained for $x$ determined
from (\ref{eq:root2}). This equation has a unique positive root $x_0$ because
the left-hand side is a falling function of $x$ that takes all positive 
values for $x\in(\omega^{1/t},\infty).$ 
This concludes the proof.
\end{proof}

This theorem implies the following result.
\begin{corollary}\label{cor:fixed} For all values of the code rate satisfying
$R\ge \log (2(1-\dgv(R))^t),$ almost all codes in the ensemble $\cC_3(t)$
approach the GV bound as $N\to\infty.$
\end{corollary}
\begin{proof}
From the previous theorem, the GV bound is met for the first time
when $x_0$ becomes 1. Substituting 1 in (\ref{eq:root2}), we obtain
a condition on $\omega$ in the form $\omega=1-2^{(R-1)/t}.$
As long as this value is less than $\dgv(R),$ the ensemble-average relative 
distance approaches $\dgv(R)$ as $N\to\infty.$
\end{proof}

We note that the condition for the attainment of the GV bound turns out
to be the same as for the ensemble $\cC_1(t)$ constructed from random
graphs. The $\epsilon$-homogeneity condition, and in particular,
the expander mixing lemma for bipartite graphs are known to approximate 
the behavior of random graphs. This approximation turns out to be good
enough to ensure that both ensembles contain GV codes in the same
interval of code rates. Moreover, for small weights the average number
of codewords for the ensemble $\cC_3(t,H)$ turns out to be smaller 
than for the ensemble $\cC_1(t).$
This is illustrated in 2 examples in Fig.~\ref{fig:spectra}.

\begin{figure}[tH]
\epsfysize=4cm
\setlength{\unitlength}{1cm}
\begin{center}\vspace*{1cm}
\begin{picture}(5,5)
\put(-5,1){\epsffile{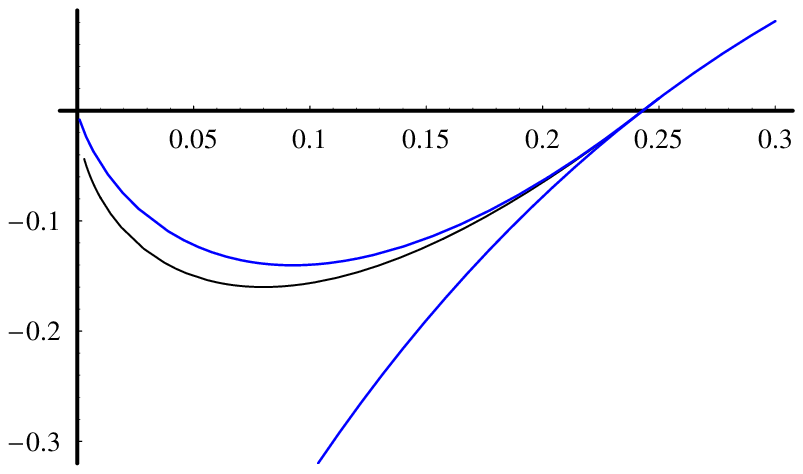}}\put(3,1){\epsffile{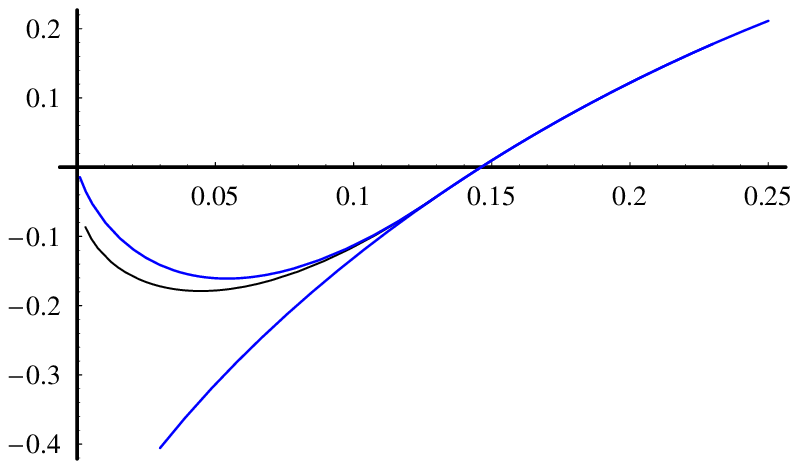}}
\put(6.5,0.5){{ (II)}}
\put(-2.9,2.3){{\scriptsize (a)}}\put(-2.9,3){{\scriptsize (b)}}
\put(4.6,2.3){{\scriptsize (a)}}\put(4.6,2.9){{\scriptsize (b)}}
\put(5.7,2.4){{\scriptsize (c)}}\put(-1.4,2.3){{\scriptsize (c)}}
\put(-1.7,0.5){{ (I)}}
\end{picture}
\vskip-5mm
\caption{Average weight spectra for ensembles of graph codes:
(I) $t=2,R=0.2$, (II) $t=3,R=0.4$; (a) ensemble $\cC_3(2,H)$, (b) ensemble 
$\cC_1(2)$, (c) ensemble of random linear codes.}
\label{fig:spectra}
\end{center}
\end{figure}

\remove{\vspace*{3ex}
\hspace*{0in}{\psfig{file=figR2t2-1.eps,width=3in,silent=}}
\hspace*{0in}{\psfig{file=figR4t3-1.eps,width=3in,silent=}}

\centerline{Fig. 1. Average weight spectra for ensembles $\cC_1,\cC_3$}
\centerline{and the ensemble of random linear codes}
\vspace*{1ex}
\hspace*{1in}(a) $t=2,R=0.2$\hspace*{2.3in}(b) $t=3,R=0.4$

\vspace*{3ex}}

For $t=2$ codes in the ensembles $\cC_3$ and $\cC_1$ 
reach the GV bound for code rates $R\le 0.202$. 
For $R>0.202$ the codes are still asymptotically good, although 
slightly below the GV bound. For these values of the rate, the 
average relative distance for the ensemble $\cC_3$ is greater than for
the ensemble $\cC_1$ as shown by the following numerical examples.

\bigskip
\begin{tabular}{lllll}
$R$ &0.3&0.5&0.7&0.9\\
$\cC_1(2)$ & 0.18558&0.09276&0.03211&0.00337\\
$\cC_3(2,H)$ &0.18605&0.09492&0.03242&0.00380
\end{tabular}

\bigskip
Similar relations between the weight spectra and distances of the ensembles
$\cC_1(t), \cC_3(t,H)$ hold also for larger values of $t$.

\section{Decoding}
For the case of a code $C(G,A)$ on a bipartite graph $G$, decoding can
be performed by a natural algorithm \cite{zem01} 
that alternates between parallel
decoding of local codes in the parts $V_1$ and $V_2$ until, hopefully,
it converges to a fixed point. In this algorithm, the most current value
of each edge (bit) is stored at the vertex in the part decoded in the most
recent iteration.
However, pursuing such an edge-oriented procedure is difficult for $t> 2.$
In \cite{bil04} the following alternative is suggested: starting from
the values of the bits stored on the edges of $H$
decode in parallel all local codes in {\em all} parts of $H$ and for 
each $v\in V$ form an independent decision about the codeword of $A$ that 
corresponds to the edges $E(v)$. Next, the values of the bits at every
vertex are updated, so that now every vertex stores an independent opinion
of its bits' values. For the update, the value of the bit $x_e(v)$ is set
to the majority value of the decoded versions of this bit at all the vertices
$v'\in e\backslash v$, where $e\ni v$ is an edge (for this to be well-defined,
the values of $t$ are assumed to be even).
The decoding then iterates, repeating this parallel
decoding round until all the vertices agree on all bits.

In \cite{bil04} this algorithm is shown to correct all patterns of errors
provided that their proportion, as a fraction of
the blocklength $N$, is less than
\begin{equation}
  \label{eq:bh}
  \binom{t-1}{t/2}^{-2/t}\left(\frac{\delta_1}{2}\right)^{(t+2)/t}-c_2(\epsilon,\delta_1,t)
\end{equation}
where $c_2(\epsilon,\delta_1,t)\rightarrow 0$ as $\epsilon\rightarrow 0$.
This algorithm consists of $\log N$ iterations, each of which has
serial running time linear in the blocklength $N$.
Its analysis relies on the $\epsilon$-homogeneous property of $H.$

For fixed values of $t>2$, if one thinks of $\delta_1$ as
a variable quantity, then the number of correctable errors in
\eqref{eq:bh} is not a constant fraction of the designed distance
\eqref{eq:dist-t}. For example, for $t=4$, \eqref{eq:bh}
gives a decoding radius equal to $N$ times the fraction
$$\frac{\delta_1^{3/2}}{2\sqrt{6}}.$$
For small $\delta_1$ this is a much smaller quantity than the 
designed distance $\delta_1^{4/3}N$. 
This consideration is reinforced by the fact
that advantages of hypergraph codes are most pronounced for small 
values of the distance $\delta.$

Our objective is to propose an alternative
decoding strategy that decodes a constant fraction of the designed
distance.

For every $i$, we shall define a {\em $i$-th subprocedure} that
decodes the subcode $A$ on
every vertex belonging to the vertex
set $V_i$. We shall claim that if the initial number of errors is less
than a bound that we shall introduce, then {\em for at least one} $i$,
the $i$-th subprocedure applied to the initial error pattern
produces a pattern with a smaller number of errors.

Let us now describe the decoding procedure in more detail.
For every vertex $v$, and the associated subspace $\{0,1\}^n$
where coordinates are indexed by the edges incident to $v$, we
will use the following {\em threshold decoding} procedure $T_\kappa$
of the constituent code $A$.
This means that we introduce a number $\kappa\geq 2$, to be optimized
later, and that we decode a vertex subcode {\em only if} its Hamming
distance to the closest codeword is less or equal to $\theta=d_1/\kappa$.
If every codeword of $A$ is at distance more than $d_1/\kappa$ we
leave the subvector untouched. Let $V_i=(v_{i,1},\dots,v_{i,m})$ be the $i$th component of $H$. 
Given an $N$-vector $z=(z(v_{i,1}),\dots,z(v_{i,m})),$ we can decode each
of the $m$ of its subvectors with $T_\kappa$, obtaining an $N$-vector
$w$. Abusing notation, we will write $w=T_\kappa(z).$
The {\em $i$-th subprocedure}
now consists of applying $T_\kappa$ to the component $V_i$. 

As mentioned above, we shall claim that one among $t$ of the $i$-th
subprocedures lowers the total number of errors. However the decoding
algorithm will not be able to discern which of the $i$-th
subprocedures is successful. So the decoder will apply all
$t$ subprocedures in parallel to the received vector, yielding $t$ 
output vectors. The next decoding iteration will have to be applied
to every output of the preceding iteration, so that $s$ iterations
of the algorithm will yield $t^s$ output vectors. We will
only apply the algorithm for a constant number of iterations however,
until we are guaranteed that the number of remaining 
error for at least one of the $t^s$ outputs has fallen below
the error-correcting capability of Bilu and Hoory's decoding
procedure. We then let the latter decoder take over and decode
all $t^s$ candidates. At least one of them is guaranteed to be the
closest codeword, and it can be singled out simply by
computing the Hamming distance of every
candidate to the initial received vector.

To give a more formal description of the algorithm, suppose that 
$y\in\{0,1\}^N$ is the vector received from the channel.
In each iteration the processing is done in parallel in all the vertices
of $H.$ Let $\sY_i^{j}=\{y_{i,l}^{(j)}\}$ 
be the set of $N$-vectors stored at the vertices
of the component $V_i$ before the $j$th iteration. By the discussion above, 
$|\sY_i^{j}(v)|\leq t^{j-1}.$ 

We begin by setting $\sY_i^1=\{y\}$ for all $i$.
Iteration $j,j=1,2,\dots,s$ consists of running $t$ parallel subprocedures.
The $i$th subprocedure applies decoder $T_\kappa$ to every vector  
$y_{i,l}^{(j)}$ in the set $\sY_i^{j},$ replacing it with the vector
$T_\kappa(y_{i,l}^{(j)}), l=1,\dots,|\sY_i^j|$. The outcome of this step creates
$t$ potentially different decodings of every vector $y_{i,l}^{(j)}\in
\sY_i^{j}, i=1,\dots, t$.
In the second part of the iteration we form the sets 
$\sY_i^{j+1}, i=1,\dots, t$ 
by replacing each vector $y_{i,l}^{(j)}\in
\sY_i^{j}$ with its decodings obtained in all
the $t$ subprocedures.

\remove{
The vertex $v$ is connected by $n$ edges to some vertices
$(v_{i_{2,1}},\dots,v_{i_{2,n}})\in V_2,
\dots,(v_{i_{t,1}},\dots,v_{i_{t,n}})\in V_t.$
Upon decoding of all the vertices of $H$, the vertex $v$ observes
the vector $y^{2}_1(v)$ obtained by the threshold decoding of $y^1(v)$
as described above plus $t-1$ vectors received from the vertices
$v_{i_{j,l}}, j=1,\dots,t,\,l=1,\dots,n$ along the edges in $E(v).$
Thus each vertex now stores $t$ vectors $y^{2}_i(v), i=1,\dots,t$.
The second round of parallel decoding will be applied to the $t$ vectors
at every vertex $v\in V$, and so on.}

Next, we prove that one of the $t$ subprocedures will actually diminish
the number of errors. This analysis also relies on $\epsilon$-homogeneity,
although in a way different from \cite{bil04}. 
Let $\cE$ be the set of coordinates, i.e. the set of edges, that
are in error. For every $i=1\ldots t$, let us partition the set
of vertices in $V_i$ that are incident to $\cE$ into three
subsets, $G_i,N_i,B_i$. The set $G_i$ is the subset of vertices
that will be correctly decoded, $N_i$ is the subset of 
vertices that are left untouched by the threshold decoder, and
$B_i$ is the set of those vertices that are wrongly decoded to a parasite
codeword of $A$. The situation is summarized in Figure~\ref{fig:Vi}. From 
now on by the {\em $\cE$-degree} of a vertex we shall mean
the degree of this vertex in the subhypergraph induced by the edge set $\cE$.
It should be clear that every vertex of $G_i$ has $\cE$-degree not more
than $d_1/\kappa$, every vertex in $N_i$ has $\cE$-degree at least
$d_1/\kappa$, and every vertex in 
$B_i$ has $\cE$-degree at least $(\kappa -1)d_1/\kappa$.

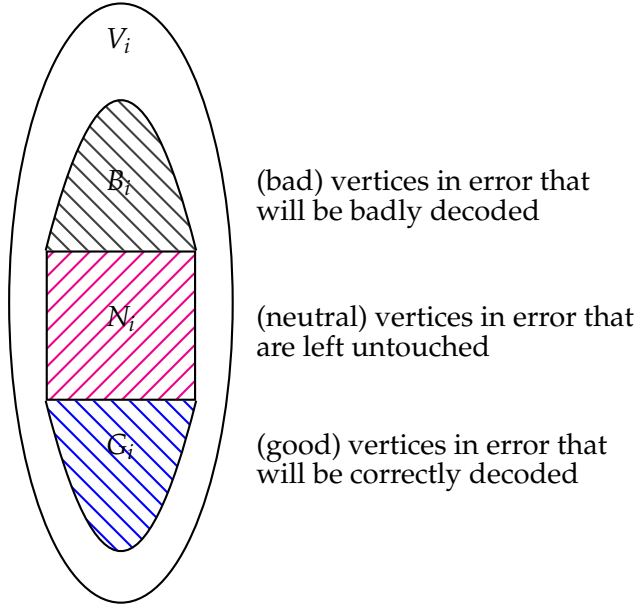
\begin{figure}
  \centering
   \begin{pspicture}(8,8.5)
\parabola[fillstyle=vlines,hatchcolor=darkgray](0.5,4.7)(1.5,6.7)

  \parabola[fillstyle=vlines,hatchcolor=blue](0.5,2.7)(1.5,0.7)

  \psframe[fillstyle=hlines,hatchcolor=magenta](0.5,2.7)(2.5,4.7)

  \psellipse(1.5,4)(1.5,4)

   \put(1.3,3.7){$N_i$}\put(1.3,2){$G_i$}\put(1.3,5.5){$B_i$}
   \put(1.3,7.4){$V_i$}
   \put(3.3,5.5){(bad) vertices in error that} 
   \put(3.3,5.1){will be badly decoded} 
   \put(3.3,3.7){(neutral) vertices in error that} 
   \put(3.3,3.3){are left untouched}
   \put(3.3,2){(good) vertices in error that} 
   \put(3.3,1.6){will be correctly decoded}
\end{pspicture}
  \caption{Details of the set of vertices incident to edges in error.
  The max $\cE$-degree in $G_i$ is less than $d_1/\kappa$,
  the min $\cE$-degree in $B_i$ is at least $(\kappa -1)d_1/\kappa$,
  the min $\cE$-degree in $N_i$ is at least $d_1/\kappa$.}
  \label{fig:Vi}
\end{figure}

We use the shorthand notation $\cE(G_i)$ to mean the set of edges
that has one of its endpoints in $G_i$. Similarly we shall write
$\cE(N_i)$ and $\cE(B_i)$. 

\begin{lemma}\label{lem:|E|/k}
  If the $i$-th decoding subprocedure introduces more errors than
it removes, then $|\cE(G_i)| \leq |\cE|/\kappa.$
Moreover, if 
    $$
      \mu_i=\frac{|\cE(N_i)|}{|\cE(N_i) \cup  \cE(B_i)|}, \quad
i=1,\dots,t
  $$
then
   $$
    |\cE(G_i)|\le\frac{1-\mu_i}{\kappa-\mu_i}|\cE|.
   $$ 
\end{lemma}

\begin{proof} The first part of the lemma follows from the second part,
which is proved as follows.
We bound from above $|\cE(G_i)|$, the set of edges removed,
by the set of edges added, $|\cE(B_i)|$: we get
  \begin{eqnarray*}
      |\cE(G_i)| & \leq & |B_i|\frac{d_1}{\kappa} =
      |B_i|d_1\left(1-\frac{1}{\kappa}\right)\frac{1}{\kappa -1}\\
      & \leq & |\cE (B_i)|\frac{1}{\kappa -1}.
    \end{eqnarray*}
The first inequality comes from the definition of $\kappa$ and the threshold
decoder. The second inequality states that $(1-1/\kappa)d_1$ is a
lower bound on the minimum $\cE$-degree in $B_i$. We now have
\begin{equation}
  \label{eq:kappa}
  |\cE|=|\cE(G_i)|+|\cE(N_i)|+|\cE(B_i)|= 
      |\cE(G_i)|+|\cE(B_i)|/(1-\mu_i)
\end{equation}
  $$
    \ge \frac{\kappa-\mu_i}{1-\mu_i}|\cE(G_i)|
  $$
which proves the lemma.
\end{proof}
 

  \begin{theorem}\label{th:e<}
    For any $\alpha >0$, if the number of errors $eN$ is such that
    \begin{equation}
      \label{eq:e<}
          e\leq(1-\alpha)\frac{\delta_1^{t/(t-1)}}{(t+1)^{(t+1)/(t-1)}}
  \end{equation}
    they can be corrected in time $O(N\log N$).
  \end{theorem}
\begin{proof}
The theorem will follow if we show that at least one subprocedure
reduces the error count by a constant fraction. Indeed, in this case
a constant number of rounds of the above algorithm will reduce the
error count to any positive proportion of the designed distance whereupon
the remaining errors will be removed in $O(\log N)$ steps of Bilu-Hoory's
algorithm.

Assume toward a contradiction that {\em all} the $i$-th decoding 
subprocedures, $i=1\ldots t$, introduce more errors than they remove.
Let us introduce the following notation:
$|\cE|=eN, S_i=B_i\cup N_i$,
     $|S_i|=\sigma_i m$.    Note that since the minimum $\cE$-degree in
     $S_i$ is at least $d_1/\kappa$, we have
     \begin{equation}
       \label{eq:ke/d}
       \sigma_i\leq \kappa e/\delta_1.
     \end{equation}
Consider the subset of edges obtained from $\cE$ by removing all edges 
incident to ``good'' vertices $G_i$ for all $i$. We are left with a 
subhypergraph $H_\cE$ with vertex set~$S_i$, $i=1\ldots t$. 
Use Lemma~\ref{lem:|E|/k} (the first part) for all $i$ to argue that the total
fraction of edges in $H_\cE$ is at least $e(1-t/\kappa)$. Applying
the $\epsilon$-homogeneous property \eqref{eq:homog} gives
  $$e\left(1-\frac{t}{\kappa}\right)\leq \sigma_1\cdots\sigma_t
    +\epsilon\min_{1\le i<j\le t}(\sigma_i\sigma_j)^{1/2}.$$
Applying \eqref{eq:ke/d} we obtain
 $$e\left(1-\frac{t}{\kappa}\right)
   \leq \left(\frac{\kappa e}{\delta_1}\right)^t+
               \epsilon\frac{\kappa e}{\delta_1}.$$
This inequality does not hold (and therefore our assumption is false)
if
  \begin{equation}\label{eq:err}
     e <
  \delta_1^{t/(t-1)}\left(\frac{1-t/\kappa-\epsilon\kappa/\delta_1}
             {\kappa^t}\right)^{1/(t-1)}.
  \end{equation}
Taking $\kappa=t+1$, rewrite the expression in the brackets on the right as
  $$
     \Big(\frac 1{t+1}\Big)^{(t+1)/(t-1)}\Big(1-\frac{(t+1)^2\epsilon}
            {\delta_1}\Big)^{\frac 1{t+1}}.
  $$
By taking sufficiently large $n$ it is possible to make $\epsilon$
small enough so that for any given $\alpha'>0$ there holds
   $$(1-{(t+1)^2\epsilon}/
            {\delta_1})^{1/{t+1}}>1-\alpha'.$$ 
This means that (\ref{eq:err}) is satisfied for all
  $$
   e<(1-\alpha')\frac{\delta_1^{t/(t-1)}}{(t+1)^{(t+1)/(t-1)}}.
  $$
Finally, choosing $\alpha'<\alpha$ guarantees that at least one subprocedure
reduces the error count by a constant fraction.
\end{proof}

\medskip

We see that the upper bound on the number of correctable errors given by
Theorem~\ref{th:e<} is a constant proportion $\gamma $ of the designed
distance $\delta N$ \eqref{eq:dist-t}, 
where $\gamma=1/(t+1)^{(t+1)/(t-1)}.$ 
For example, for $t=3,4$  we get $\gamma=1/16$ and $1/14.2$, respectively. 

The next theorem provides a better estimate of $\gamma$ by refining
the above analysis. The way this is done is to rely on the full power
of Lemma \ref{lem:|E|/k} instead of its first part as above.
  \begin{theorem}\label{th:e<2}
    For any $\alpha >0$, if the number of errors $eN$ is such that
    $$e \leq (1-\alpha)\delta_1^{t/(t-1)}\max_{\kappa\ge 2}\min_{0<\mu<1}
      f(\mu,\kappa)$$
with
  $$f(\mu,\kappa)=
   \frac{[1-t(1-\mu)/(\kappa
     -\mu)]^{1/(t-1)}}{\kappa^{t/(t-1)}[\mu+
(1-\mu)/(\kappa -1)]^{t/(t-1)}}$$
    they can be corrected in time $O(N\log N$).
  \end{theorem}
  \begin{proof}
We proceed as in the previous theorem, assuming toward a contradiction
that each subprocedure increases the error count. 
Using the definition of $\mu_i$ given above, 
  $$
    |\cE(S_i)|=\frac{|\cE(B_i))|}{1-\mu_i}=
             \frac{|\cE(N_i)|}{\mu_i}.
  $$
Recall that the subhypergraph $H_\cE$ is formed of the edges all of whose
vertices are in $S_i.$ To count the total fraction of edges $\beta(H_\cE)$ 
in the subhypergraph $H_\cE$ we employ Lemma~\ref{lem:|E|/k}:
  $$
    \beta(H_\cE)\ge e\Big(1-\sum_{i=1}^t\frac{1-\mu_i}{\kappa-\mu_i}\Big).
  $$
The $\cE$-degree of a vertex in $S_i$ (resp., $B_i$) is at least $d_1/\kappa$ 
(resp., $d_1(\kappa-1)/\kappa$). Hence
  $$
  |S_i|=|B_i|+|N_i|\le \cE(N_i)\frac\kappa{d_1}+ 
        \cE(B_i)\frac{\kappa(1-\mu_i)}{d_1(\kappa-1)}
   $$
  $$
    \leq \frac{\kappa e}{d_1}\Big(\frac{1-\mu_i}{\kappa-1}+\mu_i\Big)N.
  $$
Using the last two inequalities in (\ref{eq:homog}), we obtain
  $$
   e\Big(1-\sum_{i=1}^t\frac{1-\mu_i}{\kappa-\mu_i}\Big)
    \leq \Big(\frac{\kappa e}{\delta_1}\Big)^t\prod_{i=1}^t
          \Big(\frac{1-\mu_i}{\kappa-1}+\mu_i\Big)+\epsilon\frac{\kappa e}
               {\delta_1}.
   $$
To contradict this, let
  $$
   e<
       \Big(\frac{\delta_1}\kappa\Big)^{t/(t-1)}
        \biggl\{\frac {1-\sum_{i=1}^t\frac{1-\mu_i}{\kappa-\mu_i}-
         \epsilon\kappa/\delta_1}{ \prod_{i=1}^t
          (\frac{1-\mu_i}{\kappa-1}+\mu_i)}\biggr\}^{1/(t-1)}.
  $$
We again bound the terms that involve $\epsilon$ from below by a 
multiplicative term $1-\alpha'.$ 
Optimizing on all possible values of $\mu_i$ gives
$\mu_i=\mu$ for all $i=1\ldots t$, whereupon the expression on the right
can be replaced by $(1-\alpha)\delta_1^{t/(t-1)}f(\mu,\kappa).$
The proof is thus complete.
  \end{proof}

Numerically, the first values of the decoding radius $\rho$ given by
Theorem~\ref{th:e<2} are
$$\rho\geq \frac{\delta_1^{3/2}}{5.94}\;\;\text{for $t=3$}
\hspace{1cm}
\rho\geq \frac{\delta_1^{4/3}}{6.46}\;\;\text{for $t=4$}$$
attained for $\kappa$ satisfying $(\kappa-1)^{-t}=1-t/\kappa$ and $\mu=0$
or 1.

Can one obtain better bounds for the decoding radius~?
In principle, it is possible to obtain further improvements by introducing
{\em multiple} thresholds instead of the single decoding threshold
$\theta=d_1/\kappa,$ and approach $\rho=\delta/2$ by increasing their
number.  However we shall only be able to claim that
using one of the multiple thresholds reduces the number of errors
for one of the subprocedures, but we shall not be able to discern
which decoding threshold achieves that. This will result in yet
another layer of parallelism, further increasing the value of the
constant in the decoding complexity. We will not pursue this
line of research further here. A remaining challenge is to decode
up to half the designed distance with an iterative decoding procedure
of reasonable complexity.

\renewcommand\baselinestretch{0.9}
{\footnotesize
\providecommand{\bysame}{\leavevmode\hbox to3em{\hrulefill}\thinspace}


\begin{thebibliography}{10}
\bibitem{bar06}
A. Barg and G. Z\'emor, ``Distance properties of expander codes,''
{\em IEEE Trans. Inform. Theory}. vol. 52, no. 1, pp. 78--90, 2006.

\bibitem{bil04}
Y. Bilu and S. Hoory, ``On codes from hypergraphs,''
{\em European Journal of Combinatorics\/}, vol. 25, pp. 339--354, 2004.

\bibitem{bou99}
J.~Boutros, O.~Pothier, and G.~Z{\'e}mor, ``Generalized low density
  $(${T}anner\/$)$ codes,'' {\em Proc. IEEE ICC\/}, Vancouver, Canada,
vol.~1, pp.~441--445, 1999.

\bibitem{len99}
M.~Lentmaier and K.~Sh. Zigangirov, ``On generalized low-density
  parity-check codes based on {H}amming component codes,'' 
  {\em IEEE Communications
  Letters\/}, vol. 3, no. 8, pp. 248--260, 1999. 

\bibitem{sip96}
M.~Sipser and D.~A. Spielman, ``Expander codes,'' 
{\em IEEE Trans. Inform. Theory,} vol. 42, no. 6, pp. 1710--1722, 1996.

\bibitem{zem01}
G.~Z{\'e}mor, ``On expander codes,'' {\em IEEE Trans. Inform. Theory\/},
  vol. 47, no. 2, pp. 835--837, 2001.
\end{thebibliography}
\end{document}